\documentclass{article}

\usepackage[english]{babel}

\usepackage[letterpaper,top=2cm,bottom=2cm,left=3cm,right=3cm,marginparwidth=1.75cm]{geometry}

\usepackage{amsthm,amssymb,amsmath}  
\usepackage{mathtools}
\usepackage{graphicx}
\usepackage[colorlinks=true, allcolors=blue]{hyperref}
\usepackage[capitalize]{cleveref}
\usepackage{xspace}
\usepackage{thmtools}
\usepackage{thm-restate}
\usepackage{url}
\usepackage{authblk}

\def\dd{\mathinner{.\,.}} 
\newcommand{\cO}{\mathcal{O}}
\newcommand{\Occ}{\textsf{Occ}}
\newcommand{\PrefSuf}{\textsf{PrefSuf}}

\newtheorem{theorem}{Theorem}
\newtheorem{fact}{Fact}
\newtheorem{claim}{Claim}
\newtheorem{lemma}{Lemma}
\newtheorem{remark}{Remark}
\newtheorem{corollary}{Corollary}

\title{Optimal prefix-suffix queries with applications}
\author[1,2]{Solon P.\ Pissis}
\affil[1]{CWI, Amsterdam, The Netherlands}
\affil[2]{Vrije Universiteit, Amsterdam, The Netherlands}

\begin{document}

\maketitle

\begin{abstract}
We revisit the classic \emph{border tree} data structure [Gu, Farach, Beigel, SODA 1994] that answers the following \emph{prefix-suffix} queries on a string $T$ of length $n$ over an integer alphabet $\Sigma=[0,\sigma)$: for any $i,j \in [0,n)$ return all occurrences of $T$ in $T[0\dd i]T[j\dd n-1]$. The border tree of $T$ can be constructed in $\cO(n)$ time and answers prefix-suffix queries in $\cO(\log n + \Occ)$ time, where $\Occ$ is the number of occurrences of $T$ in $T[0\dd i]T[j\dd n-1]$. Our contribution here is the following. We present a completely different and remarkably simple data structure that can be constructed in the optimal $\cO(n/\log_\sigma n)$ time and supports queries in the optimal $\cO(1)$ time. Our result is based on a new structural lemma that lets us encode the output of any query in \emph{constant time and space}. We also show a new direct application of our result in pattern matching on node-labeled graphs.
\end{abstract}


\section{Introduction}
Let $T=T[0\dd n-1]=T[0\dd n)$ be a string of length $n$ over an integer alphabet $\Sigma=[0,\sigma)$ with $\sigma=n^{\cO(1)}$.
We would like to preprocess $T$ in order to answer the following type of queries: for any $i,j \in [0,n)$, return all occurrences of $T$ in $T[0\dd i]T[j\dd n-1]$. This type of query, which we denote here by $\PrefSuf(i,j)$, is supported by the classic \emph{border tree} data structure in $\cO(\log n + \Occ)$ time after an $\cO(n)$-time preprocessing~\cite{DBLP:conf/soda/GuFB94}. We show the following optimal data structure for $\PrefSuf(i,j)$ queries.

\begin{theorem}\label{the:main}
    For any string $T$ of length $n$ over an alphabet $\Sigma=[0,\sigma)$ with $\sigma=n^{\cO(1)}$,
    we can answer $\PrefSuf(i,j)$ queries, for any $i,j\in[0,n)$, in $\cO(1)$ time after an $\cO(n/\log_\sigma n)$-time preprocessing. The data structure size is $\cO(n/\log_\sigma n)$ and the output is given as a compact representation of $\cO(1)$ size.
\end{theorem}

In the word RAM model with $w$-bit machine words and $w=\Omega(\log n)$, 
$T$ is represented as an array: each letter occupies one
machine word. However, a single letter can be represented using $\lceil \log \sigma \rceil$ bits (i.e., packed representation), which could be (significantly) less than $w$ (e.g., for a constant-sized alphabet). Thus, in the word RAM model, it takes $\cO(n/\log_\sigma n)$ words to store $T$ and $\cO(n/\log_\sigma n)$ time to read it. Hence \cref{the:main} is \emph{optimal} with respect to the construction time and the query time. 

\paragraph{Border tree.} We start with an informal description of the border tree data structure.
The classic KMP algorithm~\cite{DBLP:journals/siamcomp/KnuthMP77} constructs an automaton over string $T=T[0\dd n-1]$. The automaton consists of an initial state and: (1) one state per prefix of $T$ numbered from $0$ to $n-1$; 
(2) a \emph{success} transition from state $i$ to state $i+1$; and (3) a \emph{failure} transition from state $i$ to the state representing the longest string that is both a prefix and a suffix (known as \emph{border}) of $T[0\dd i]$.
The failure transitions form the \emph{failure tree}. Thus any path in the failure tree from the root to a state $i$
specifies all the borders of $T[0\dd i]$. Since the number of distinct borders of $T[0\dd i]$ can be $\Theta(n)$, we apply a grouping of the borders based on periodicity that results in the \emph{border tree}: a compacted version of the failure tree with $\cO(n)$ states and $\cO(\log n)$ depth. The construction time is $\cO(n)$ and the size of the data structure is $\cO(n)$. To answer a \emph{prefix-suffix} query $\PrefSuf(i,j)$, we use the border tree $\mathcal{T}$ of $T$ and the border tree $\mathcal{T}^R$ of $T^R=T[n-1]\ldots T[0]$. Given that $T$ is of a fixed length $n$ and the border tree has $\cO(\log n)$ depth, we need to check $\cO(\log n)$ pairs of states: one from $\mathcal{T}$ and one from $\mathcal{T}^R$. For each pair, we solve one linear equation in $\cO(1)$ time to check for the length constraints. The query time is $\cO(\log n + \Occ)$.

Gu, Farach, and Beigel~\cite{DBLP:conf/soda/GuFB94} used prefix-suffix queries in their dynamic text indexing algorithm to efficiently locate the occurrences of a pattern spanning an edit operation in the text by maintaining the longest prefix and the longest suffix of the pattern occurring right before and right after the edit, respectively. Since its introduction~\cite{DBLP:conf/soda/GuFB94}, the border tree has been used for several pattern matching tasks (e.g.,~\cite{DBLP:journals/jal/Ferragina97,DBLP:journals/jcss/AmirBF96,DBLP:journals/talg/AmirLLS07,DBLP:conf/wabi/Ascone0CEGGP24}).

\paragraph{Our contribution.} We present a completely different and remarkably simple data structure that can be constructed in the optimal $\cO(n/\log_\sigma n)$ time in the word RAM model and answers $\PrefSuf(i,j)$ queries in the optimal $\cO(1)$ time. We remark
that it is quite standard to represent the set of occurrences of a string $X$ in another string $Y$ with $|Y|<2|X|$ in $\cO(1)$ space due to the following folklore fact:

\begin{fact}[\cite{DBLP:conf/icalp/PlandowskiR98}]
Let $X$ and $Y$ be strings with $|Y| < 2|X|$. The set of occurrences (starting positions) of $X$ in $Y$ forms a single arithmetic progression.    
\end{fact}

However, to the best of our knowledge, the set of occurrences of $T$ in $T[0\dd i]T[j\dd n-1]$, \textbf{for all $i,j\in[0,n)$}, has not been characterized before. To arrive at \cref{the:main}, we prove a structural lemma that lets us encode the output of $\PrefSuf(i,j)$ queries for any $i,j\in[0,n)$ in \emph{constant time and space}. The border tree, and, in particular prefix-suffix queries, have been mainly used in dynamic text indexing algorithms, which however involve many other crucial primitives to arrive at their final query time. We show here instead a new direct application of our data structure in pattern matching on node-labeled graphs, a very active topic of research~\cite{DBLP:journals/talg/EquiMTG23,DBLP:journals/jacm/CotumaccioDPP23,DBLP:conf/wabi/Ascone0CEGGP24,DBLP:conf/cpm/AlankoCCKMP24}. In particular,
we formalize \emph{bipartite pattern matching} as a core problem that underlies any algorithm for pattern matching on node-labeled graphs. 
For intuition, consider two nodes $u$ and $v$ in a directed graph where nodes are labeled by \emph{strings}. Some suffixes of node $u$ match some prefixes of a given pattern $P$, and some prefixes of node $v$ match some suffixes of $P$. 
We would like to have a data structure, constructed over $P$, that takes $\cO(1)$ time to process the directed edge $(u,v)$. Namely, in this setting, we find all occurrences of $P$ spanning \emph{at most two nodes} of the graph. (It is trivial to find the occurrences of $P$ in a single node using any linear-time pattern matching algorithm~\cite{DBLP:journals/siamcomp/KnuthMP77}.)
Indeed, the bipartite pattern matching problem has been (implicitly) introduced by Ascone et al.~\cite{DBLP:conf/wabi/Ascone0CEGGP24} for pattern matching on \emph{block graphs} (a restricted version of node-labeled graphs~\cite{DBLP:conf/wabi/MakinenCENT20,DBLP:journals/algorithmica/EquiNACTM23,DBLP:journals/tcs/RizzoENM24}), and the border tree~\cite{DBLP:conf/soda/GuFB94} was used to solve it. If we apply \cref{the:main} on $P$, we can answer \emph{any such query in $\cO(1)$ time} instead of $\cO(\log |P| + \Occ)$ time.

\paragraph{Paper organization.} In \cref{sec:solution}, we present the proof of \cref{the:main}.
In \cref{sec:app}, we present the application of \cref{the:main} on bipartite pattern matching.
We conclude this paper in \cref{sec:fin}.

\section{Optimal prefix-suffix queries}\label{sec:solution}

For a fixed pair $(i,j)$, we set $T'=T[0\dd i]T[j\dd n-1]$. In particular, $T'$ is a string of length $n+(i-j+1)$.
Let us remark that the relevant cases for $\PrefSuf(i,j)$ are all $i,j\in[0,n)$ such that $i \geq j$, otherwise the answer to the query is trivial: if $j=i+1$, then we have a single occurrence of $T$ starting at position $0$ of $T'$. If $j>i+1$, then we have no occurrence of $T$ in $T'$ because $|T'|<n$.
Hence we assume $i \geq j$.

We start with a few standard definitions on strings from~\cite{DBLP:books/daglib/0020103}. An integer $p>0$ is a \emph{period} of a string $P$ if $P[i] = P[i + p]$, for all $i \in [0,|P|-p)$. The smallest period of $P$ is referred to as \emph{the period} of $P$ and is denoted by $\textsf{per}(P)$. A string $P$ is called \emph{periodic} if $\textsf{per}(P)\leq |P|/2$. A \emph{border} $B$ of a string $P$ is a string of length $|B|<|P|$ so that $B$ occurs as both a prefix and a suffix of $P$.
\cref{fct:per}, stating that the notion of border and the notion of period are dual, is well-known.

\begin{fact}[\cite{DBLP:books/daglib/0020103}]\label{fct:per}
For every period $p$ of a string $P$, there is a border of length $|P|-p$. In particular, the longest border of $P$ is of length $|P|-\textsf{per}(P)$.
\end{fact}

The periods of any string $P$ also satisfy the following well-known lemma.

\begin{lemma}[Periodicity lemma \cite{periodicity}]\label{lem:periodicity}
If $p$ and $q$ are both periods of a string $P$ and satisfy
$p + q \leq |P|$, then $\gcd(p, q)$ is a period of $P$.
\end{lemma}

We proceed by showing the main structural lemma of this work.

\begin{lemma}\label{lem:main_lemma}
For any string $T$ of length $n$ and any $i,j\in[0,n)$ such that $i \geq j$, $T[0\dd i]T[j\dd n-1]$ has an occurrence of $T$ that is neither a prefix nor a suffix of  $T[0\dd i]T[j\dd n-1]$ if and only if $T$ is periodic with period $p=\textsf{per}(T)<i-j+1$ and $(i-j+1) \mod p = 0$.
\end{lemma}
\begin{proof}
Let us set $T'=T[0\dd i]T[j\dd n-1]$ with $|T'|=n+(i-j+1)$. Inspect~\cref{fig:main_lemma} for an illustration.

($\Rightarrow$) For the forward implication, assume $T'[j'\dd i']=T$ with $j'>0$ and $i'<|T'|-1$; namely, $T$ has an occurrence in $T'$ that it is neither a prefix nor a suffix of $T'$. By hypothesis, $|T[0\dd i]|\leq n$ and $|T[j\dd n)|\leq n$.
Since $T[0\dd i]=T'[0\dd i]=T'[j'\dd i+j']$, by \cref{fct:per}, $T'[0\dd i+j']$ has period $j'$ (in \cref{fig:main_lemma}, $T'[0\dd i]=T'[j'\dd i+j']=\texttt{aabaabaaba}$ and $j'=3$). Symmetrically, since $T[j\dd n)=T'[i+1\dd |T'|)=T'[i+1-(|T'|-i'-1)\dd i']$, by \cref{fct:per}, $T'[i+1-(|T'|-i'-1)\dd |T'|)$ has period $|T'|-i'-1$ (in \cref{fig:main_lemma}, $T'[i+1\dd |T'|)=T'[i+1-(|T'|-i'-1)\dd i']=\texttt{abaabaaba}$ and $|T'|-i'-1=3$). 

Note that because $i\geq j$, string $T'[0\dd i]=T[0\dd i]$ has a nonempty string $F$ of length $i-j+1$ as a suffix \emph{and} string $T'[i+1\dd |T'|)=T[j\dd n)$ has the same string $F$ as a prefix (in \cref{fig:main_lemma}, $F=\texttt{abaaba}$). 

\begin{claim}\label{clm:len}
$|F|=j' + (|T'|-i'-1)$.    
\end{claim}
\begin{proof}
We have $|T'|-i'-1+j'=n+(i-j+1)-i'-1+j'$ by substituting the length of $T'$.
Then $n+(i-j+1)-i'-1+j'=n+(i-j+1)-n=i-j+1=|F|$ because $|T'[j'\dd i']|=n$.
\end{proof}

By the above discussion and \cref{clm:len}, string $F$ has \emph{both} periods $j'$ and $|T'|-i'-1$, and so we can apply \cref{lem:periodicity} to get that
$F$ has period $\gcd(j', |T'|-i'-1)$,
which is also a divisor of $|F|=i - j + 1$,
and thus that the whole $T'[j'\dd i']=T$ is periodic with period $\gcd(j', |T'|-i'-1)$.
In particular, $\gcd(j', |T'|-i'-1)$ is either equal to the period $p=\textsf{per}(T)$ or its multiple. Finally, since $j'>0$, $i'<|T|-1 \implies |T'|-i'-1>0$, and $|F|=j' + (|T'|-i'-1)$, it also follows that $\textsf{per}(T)<|F|$.

($\Leftarrow$) For the reverse implication, if $T$ is periodic with period $p=\textsf{per}(T)<i-j+1$ and $(i-j+1) \mod p = 0$, then $T'=T[0\dd i]T[j\dd n-1]$ is also periodic, and $T$ occurs at positions $0,p,2p,\ldots$ of $T'$. In particular, $T$ has at least three occurrences in $T'$ because $p<i-j+1$ and $|T'|=n+(i-j+1)$.
\end{proof}

\begin{figure}
    \centering
    \includegraphics[width=0.4\linewidth]{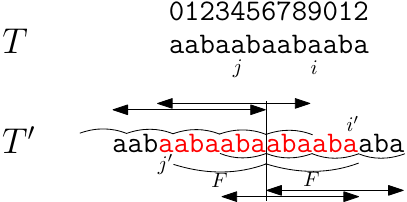}
    \caption{Illustration of \cref{lem:main_lemma} for $T=\texttt{aabaabaabaaba}$ ($n=13$), $i=9$, and $j=4$. We colored red an occurrence $T'[j'\dd i']=T$ with $j'=3>0$ and $i'=15<|T'|-1$. The pair of lines with arrows on the top of $T'$ show that $T'[0\dd i+j']$ has period $j'=3$. The pair of lines with arrows at the bottom of $T'$ show analogously that $T'[i+1-(|T'|-i'-1)\dd |T'|)$ has period $|T'|-i'-1=3$. We further have that $|F|=j' + (|T'|-i'-1)=3+3=6$ and that $T'[j'\dd i']=T$ is periodic with period $\textsf{per}(T)=3$.}
    \label{fig:main_lemma}
\end{figure}

The following corollary follows immediately from \cref{lem:main_lemma}.

\begin{corollary}\label{cor:compact}
    If string $T$ is periodic with period $p=\textsf{per}(T)<i-j+1$ and $(i-j+1) \mod p = 0$, with $i \geq j$, the set of occurrences of $T$ in $T'=T[0\dd i]T[j\dd n-1]$ forms a single arithmetic progression with difference $p$.
\end{corollary}

\paragraph{Data structure.} Let us recall that $T$ is a string of length $n$ over an integer alphabet $\Sigma=[0,\sigma)$ with $\sigma=n^{\cO(1)}$.
If string $T$ is periodic, the smallest period $p=\textsf{per}(T)$ of string $T$ can be computed in $\cO(n/\log_\sigma n)$ time~\cite{DBLP:journals/corr/KociumakaRRW13,DBLP:journals/tcs/Ben-KikiBBGGW14}. The algorithms of~\cite{DBLP:journals/corr/KociumakaRRW13,DBLP:journals/tcs/Ben-KikiBBGGW14} return $p=\perp$ if $T$ is not periodic within the same time complexity.
The first part of the construction algorithm is to compute this $p$.
We also construct a \emph{longest common extension} (LCE) data structure over $T$.
The LCE data structure answers the following type of queries: for any $i,j\in[0,n)$,
return the length of the longest common prefix (\textsf{LCP}) of $T[i\dd n)$ and $T[j\dd n)$. By reversing $T$, we can answer the symmetric longest common suffix (\textsf{LCS}) queries. \textsf{LCP} and \textsf{LCS} queries can be answered in $\cO(1)$ time after an $\cO(n/\log_\sigma n)$-time preprocessing by constructing the LCE data structure over $T$ from~\cite{DBLP:conf/stoc/KempaK19}. The total construction time and the size of our data structure are in $\cO(n/\log_\sigma n)$.
Notably, this is $\cO(n/\log n)$ for any constant-sized alphabet.

\paragraph{Answering prefix-suffix queries.} Here comes a $\PrefSuf(i,j)$ query.
We simply apply \cref{lem:main_lemma}. If $T$ is periodic with period $p\neq \perp$, we also check whether $p<i-j+1$ and $(i-j+1) \mod p = 0$. 
These checks are performed in $\cO(1)$ time.
If the answer is positive, we return all occurrences in $\cO(1)$ time by applying \cref{cor:compact}. If the answer is negative or $T$ is aperiodic, we trigger two LCE queries. Namely, if $\textsf{LCP}(T[i+1\dd n),T[j\dd n))= n-i-1$, then $T$ occurs as a prefix of $T'=T[0\dd i]T[j\dd n-1]$;
if $\textsf{LCS}(T[0\dd j-1],T[0\dd i])= j$ then $T$ occurs as a suffix of $T'$. 
The total query time is thus $\cO(1)$.

\begin{remark}
It is possible to have an aperiodic string $T$ that occurs as both a prefix \emph{and} a suffix of $T'=T[0\dd i]T[j\dd n-1]$ (we handle this with the two LCE queries). For instance, let $T=\texttt{aababaab}$, $i=5$, and $j=1$.
$T$ is aperiodic and it occurs as both a prefix and a suffix in
$T'=\texttt{aababaababaab}$.
\end{remark}

The algorithm is correct by \cref{lem:main_lemma}.
We have arrived at \cref{the:main}.

\section{Application: Bipartite pattern matching}\label{sec:app}

Pattern matching on node-labeled graphs is an old topic~\cite{ManberWu,DBLP:conf/cpm/Akutsu93,DBLP:conf/cpm/ParkK95,DBLP:journals/jal/AmirLL00} that is regaining a lot of attention~\cite{DBLP:journals/talg/EquiMTG23,DBLP:journals/jacm/CotumaccioDPP23,DBLP:conf/wabi/Ascone0CEGGP24,DBLP:conf/cpm/AlankoCCKMP24} due to its application in computational pangenomics~\cite{DBLP:journals/nc/BaaijensBBVPRS22,DBLP:journals/bioinformatics/GuarracinoHNPG22,DBLP:journals/bioinformatics/GarrisonG23}. 
A \emph{pangenome} is any collection of genomic sequences to be analyzed jointly or to be used as a reference~\cite{DBLP:journals/bib/Consortium18}; e.g., a collection of highly-similar genomes of the same species. A standard pangenome representation is that of a \emph{variation graph}~\cite{DBLP:journals/nc/BaaijensBBVPRS22}: a directed graph $G(V,E,W)$, whose nodes are labeled by nonempty strings, equipped with a set $W$ of distinguished walks corresponding to variants that we want to retain in the representation (inspect \cref{fig:variation_graph}, middle). The pattern matching problem on a variation graph is then defined naturally as the problem of finding occurrences of a pattern $P$ of length $m$ spanning one or more nodes that lie on a walk from $W$. It should be clear that any algorithm solving this general problem \emph{must} also solve the restricted version of it asking for occurrences spanning up to two nodes. 

\begin{figure}[ht]
    \centering
    \includegraphics[width=0.85\linewidth]{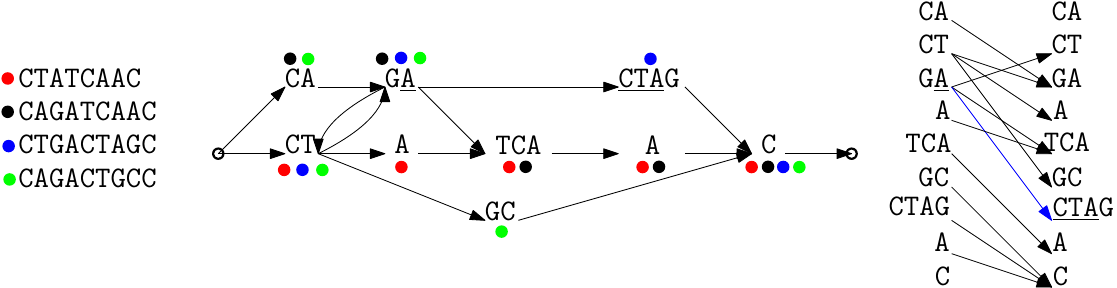}
    \caption{Illustration of a variation graph and of bipartite pattern matching. The pattern $P=\texttt{ACTA}$ has an occurrence (underlined) spanning two nodes lying on a valid walk for the sequence marked blue.}
    \label{fig:variation_graph}
\end{figure}

This restricted version is also interesting theoretically as there exists a conditional lower bound for the general problem of pattern matching on node-labeled graphs: for any constant $\epsilon > 0$, an $\cO(|E|^{1-\epsilon} m)$-time or an $\cO(|E|m^{1-\epsilon})$-time algorithm, even when $P$ is a binary string, cannot be achieved unless the Strong Exponential Time Hypothesis fails~\cite{DBLP:journals/talg/EquiMTG23}. We formalize the restricted version of the problem on a bipartite graph (inspect \cref{fig:variation_graph}, right) and call it the \emph{bipartite pattern matching} problem.

Let $G(U,V,E)$ be a bipartite graph such that $U$ is a set 
of nodes, each labeled by a string; $V$ is a set 
of nodes, each labeled by a string; and $E$ is a set of edges.
Let $N_U$ and $N_V$ denote the total length of the labels in $U$ and $V$, respectively.
We are also given a string $P$ (the pattern) of length $m$. The \emph{bipartite pattern matching} problem asks us to find every pair of nodes $(u,v)$, such that there exists 
a node $u\in U$ with a label having suffix $s$, 
a node $v\in V$ with a label having prefix $p$, $P=s\cdot p$, and $(u,v) \in E$. 
Let $N=N_U+N_V$. Again, we make the standard assumption of an integer alphabet of size polynomial in $N$. We solve this problem in $\cO(N + |E|)$ time after preprocessing $P$ in $\cO(m)$ time.

In preprocessing, we construct: the \emph{KMP automaton}~\cite{DBLP:journals/siamcomp/KnuthMP77} of $P$;
the KMP automaton of $P^R$, the reverse of $P$; 
and the data structure of \cref{the:main} on $P$ in $\cO(m)$ total time.
For each node $u$ in $U$, we find the longest suffix of the label of $u$ that is a prefix of $P$.
This takes $\cO(N_U)$ time using the KMP automaton of $P$. 
For each node $v$ in $V$, we find the longest prefix of the label of $v$ that is a suffix of $P$.
This takes $\cO(N_V)$ time using the KMP automaton of $P^R$.
By \cref{the:main}, we spend $\cO(1)$ time per edge $(u,v)$: we ask a single prefix-suffix query for the prefix of $P$ and the suffix of $P$ found above. Our solution not only finds the relevant pair of nodes: it encodes all occurrences spanning each such pair.

\begin{theorem}\label{the:bpm}
    The bipartite pattern matching problem can be solved in $\cO(N + |E|)$ time after preprocessing the pattern in $\cO(m)$ time.
\end{theorem}

\section{Concluding remarks}\label{sec:fin}
The purpose of this paper is twofold. From the results perspective, \cref{the:main} may find other applications in pattern matching, and \cref{the:bpm} may serve as a useful tool for developing practical algorithms for pattern matching in labeled graphs or as inspiration for future theoretical developments. From an educational perspective, our paper may serve as material for a class on advanced data structures showcasing a remarkably simple and optimal data structure obtained through combinatorial insight. 

\section*{Acknowledgments}
This work was supported by the PANGAIA and ALPACA projects that have received funding from the European Union’s Horizon 2020 research and innovation programme under the Marie Skłodowska-Curie grant agreements No 872539 and 956229, respectively. I am grateful to Wiktor Zuba for many helpful discussions.

\bibliographystyle{plain}
\bibliography{references}

\end{document}